\documentclass[parskip=half,DIV=11,11pt]{scrartcl}
\pdfoutput=1
\usepackage{cmap}
\usepackage[utf8]{inputenc}
\usepackage[T1]{fontenc}
\usepackage[ngerman,english]{babel}
\usepackage{csquotes}
\usepackage{lmodern}
\usepackage{microtype}
\usepackage{tikz}
\usetikzlibrary{positioning,arrows,shapes,calc}
\usepackage{pgfplots}
\usepackage{amsmath}
\usepackage{amssymb}
\usepackage{mathtools}
\usepackage[thmmarks,hyperref,amsmath]{ntheorem}
\usepackage{textcomp}
\usepackage[output-decimal-marker={.},
  number-unit-product={\,},range-phrase=--]{siunitx}
\usepackage[labelfont={sf,bf},width=0.9\textwidth,
            justification=justified,skip=6pt]{caption}
\usepackage{newfloat}
\usepackage{varwidth}
\usepackage[section]{placeins}
\usepackage{subcaption}
\usepackage{basicmath}
\usepackage[backend=bibtex,maxcitenames=2,maxbibnames=10,sorting=none]{biblatex}
\usepackage{setspace}
\usepackage{enumitem}
\usepackage[noend]{algpseudocode}

\usepackage{rotating}
\usepackage{booktabs}
\usepackage{multirow}
\usepackage{multicol}
\usepackage[]{hyperref}

\theoremstyle{plain}
\theoremseparator{:}
\theoremheaderfont{\normalfont\sffamily\bfseries}
\theorembodyfont{\upshape}
\theoremsymbol{\quad\ensuremath{\diamond}}
\newtheorem{theorem}{Theorem}[section]

\newtheorem{corollary}[theorem]{Corollary}
\newtheorem{lemma}[theorem]{Lemma}

\newtheorem{definition}[theorem]{Definition}

\newtheorem{example}[theorem]{Example}
\theoremsymbol{\quad\ensuremath{\square}}
\theoremstyle{nonumberplain}
\theoremheaderfont{\bfseries\itshape}
\newtheorem{proof}{Proof}

\newcommand{\thmbreak}{\unskip\hfil\penalty0\hfilneg\hskip1em plus1em minus.5em\relax}

\pgfplotsset{compat=1.14}
\newcommand{\kwd}{\emph}

\title{Formulating Beurling LASSO for Source Separation via Proximal Gradient Iteration}
\author{Sören Schulze\footnote{Correspondence: \href{mailto:sschulze@uni-bremen.de}{sschulze@uni-bremen.de}}\\Center for Industrial Mathematics\\University of Bremen
  \and
  Emily J. King\\Mathematics Department\\Colorado State University}
\date{February 16, 2022}

\begin{document}
\maketitle

\begin{abstract}
  Beurling LASSO generalizes the LASSO problem to finite Radon measures
  regularized via their total variation.
  Despite its theoretical appeal, this space
  is hard to parametrize, which poses an algorithmic challenge.
  We propose a formulation of continuous convolutional source separation
  with Beurling LASSO that avoids the explicit computation of the
  measures and instead employs the duality transform of the proximal mapping.
\end{abstract}

\section{Introduction}

LASSO is a finite-dimensional least-squares problem that is regularized via
the $1$-norm \cite{Tibshirani96}.
In the \kwd{basis pursuit denoising} formulation
\parencite[cf.][Section 3.1]{Foucart13}, it can be written as:
\begin{equation}\label{eq:bpd}
  v=\argmin_{v\in\R^n} \frac{1}{2}\snorm{Av-b}_2^2+\alpha\snorm{v}_1,\qquad \alpha>0,
\end{equation}
where $A\in\R^{m\times n}$ and $b\in\R^m$.  Algorithmically, it is usually
solved via \kwd{proximal gradient} methods
\parencite[cf.][Section~4.2]{Parikh14}, often called ISTA in that setting,
or an accelerated version called FISTA \parencite{Beck09}.

A shortcoming of \eqref{eq:bpd} is that $v$ is discrete.  Therefore, while
it is possible to represent convolutions in this formulation such as
via \kwd{convolutional sparse coding} \parencite{Bristow13}, those are
necessarily limited to a predetermined grid.  An approach to off-the-grid
convolutions was proposed by \textcite{Ekanadham11} as
\kwd{continuous basis pursuit}; it still uses a grid but interpolates
between the points.

Beurling LASSO avoids the discretization of the solution altogether by
operating on measures, and it is not limited to convolutions.
However, it is hard to represent such measures parametrically, which
poses a problem in practical applications.  For the purposes of
source separation, however, we do not need explicit access to the
full solution; instead, we present an approach to obtain the separated sources
via proximal gradient iteration.

\section{Beurling LASSO}

\kwd{Beurling LASSO} (BLASSO\index{BLASSO|see{Beurling LASSO}})
\parencites{Castro12}{Bredies13}{Catala17}[cf.][]{Poon19}
is a variation of basis pursuit denoising \eqref{eq:bpd}
where the solution is a finite \kwd{Radon measure}.  Formulations of
basis pursuit denoising in infinite-dimensional vector spaces
are often called \kwd{Tikhonov}-regularized problems
\parencite[cf.][Chapter~4]{Schuster12}.  We first state:
\[
v=\argmin_{v\in\mathcal{X}} \frac{1}{2}\snorm{Av-b}_{\mathcal{H}}^2+\alpha\snorm{v}_{\mathcal{X}},\qquad \alpha>0,
\]
where $\mathcal{X}$ is a real Banach space and $\mathcal{H}$ is a real
Hilbert space, $A\colon\mathcal{X}\to\mathcal{H}$ is a
continuous linear operator, and $b\in\mathcal{H}$.

Banach spaces can be very general; for instance, $C_0(\R)$
is a Banach space with the norm
\[
\snorm{x}_\infty=\max_{\omega\in\R}\sabs{x(\omega)}.
\]
Via the Riesz representation theorem \parencite[cf.][{}6.19]{Rudin87},
it follows that its dual space $\mathcal{M}(\R)$
is that of finite regular signed Borel measures,
also called finite Radon measures \parencite[cf.][Section~2]{Poon19}.
It becomes a Banach space when equipped with the norm of \kwd{total variation}~(TV\index{TV|see{total variation}}):
\[
\norm{\nu}_{\mathrm{TV}}=\sabs{\nu}(\R)=\sup_{x\in C_0(\R)}\biggl\{\int x\:\dif\nu:\snorm{x}_{\infty}\leq 1\biggr\},\qquad
\nu\in\mathcal{M}(\R),
\]
which is the dual norm of $\snorm{\cdot}_\infty$.

Since $C_0(\R)\supset\mathcal{S}(\R)$ (where $\mathcal{S}(\R)$ is the
Schwartz space over $\R$ \cite[cf.][{}7.3]{Rudin91}), it follows that
$\mathcal{M}(\R)=C_0'(\R)\subset\mathcal{S}'(\R)$,
and therefore, any finite Radon measure can also be regarded as a distribution.

\begin{example}
  \parencite[cf.][Sections 1--3]{Poon19}\hskip1em plus1em minus.5em
  Let $X\in\mathcal{B}(\R)$, where $\mathcal{B}(\R)$ is the
  Borel $\sigma$-algebra over~$\R$.
  \begin{itemize}
  \item The \kwd{Dirac measure} is defined as:
    \[
    \delta(X)=
    \begin{cases}
      1,&\text{if $0\in X$},\\
      0,&\text{otherwise},
    \end{cases}
    \]
    and it can be identified with the Dirac $\delta$-distribution.
  \item The Dirac measure can be translated, and linear combinations
    of translated Dirac measures can be summed.  In fact,
    for any series $(c_j)\in\ell_1(\Z)$, we can define:
    \[
    \nu_c(X)=\sum_{j\in\Z}c_j\,\chi_{j\in X},
    \]
    such that $\snorm{\nu_c}_{\mathrm{TV}}=\snorm{(c_j)}_{\ell_1}$.
    Here, $\chi_{j\in X}=1$ if $j\in X$ and $0$ otherwise.
  \item For any function $f\in L_1(\R)$, we can define:
    \[
    \nu_f(X)=\int_X f\:\dif\lambda,
    \]
    where $\lambda$ is the Lebesgue measure.  Then
    $\snorm{\nu_f}_{\mathrm{TV}}=\snorm{f}_{L_1}$, which may seem surprising
    since the TV-norm on functions has a different definition and
    can also be used as a regularizer \parencite[cf.][Section 8.1]{Vogel02}.
    However, when considering
    \[
    F(t)=\int_{-\infty}^t f(\omega)\:\dif\omega\qquad
    \text{such that}\qquad
    \nu_f\bigl((a,b])=F(b)-F(a),
    \]
    it follows that $\snorm{F}_{\textrm{TV}}=\snorm{\nu_f}_{\textrm{TV}}$
    \parencite[cf.][Theorem 3.29]{Folland99}.
  \end{itemize}
  With these identifications, we can conclude that both
  $\ell_1(\Z)\subset\mathcal{M}(\R)$
  and $L_1(\R)\subset\mathcal{M}(\R)$.
\end{example}

We now state the BLASSO problem as:
\begin{equation}\label{eq:blasso}
\min_{\nu\in\mathcal{M}(\R)}\biggl[\frac{1}{2}\snorm{A\nu-b}^2_{\mathcal{H}}+\alpha\snorm{\nu}_{\mathrm{TV}}\biggr],\qquad\alpha>0.
\end{equation}
The immediate difficulty is that while the space $\mathcal{M}(\R)$
is versatile, it is also hard to parametrize.  It can be shown that
under certain assumptions, the TV-norm induces sparsity such that
the solution $\nu$ is a finite linear
combination of translated Dirac measures \parencite[Section~4.1]{Bredies20}.
A solver which makes explicit use of this representation is the
\kwd{sliding Frank-Wolfe} algorithm \parencite{Denoyelle19} which employs a
non-convex solver (such as BFGS) in order to refine the shifts and amplitudes.

\section{The Dual Problem}

Another approach to solving the problem \eqref{eq:blasso} is to transform
it in such a way that the \enquote{problematic} space $\mathcal{M}(\R)$
does not need to be explicitly handled anymore.  For this, we need
some elements of convex analysis:

\begin{definition}\label{def:convex}
  \parencites[cf.][{}4.2]{Rudin91}[Sections 2.3, 2.4]{Zalinescu02}\thmbreak
  Let $\mathcal{X}$ be a real topological vector space
  vector space and $x\colon\mathcal{X}\to\overline{\R}$.
  \begin{itemize}
  \item The topological vector space $\mathcal{X}^*$ is defined such that
    either $\mathcal{X}^*=\mathcal{X}'$ or $(\mathcal{X}^*)'=\mathcal{X}$, where $\mathcal{X}'$ is the topological dual space of $\mathcal{X}$, and
    $(\mathcal{X}^*)'$ is the topological dual space of $\mathcal{X}^*$.
    For $\omega\in\mathcal{X}$ and $\omega^*\in\mathcal{X}^*$, we note the
    \kwd{dual pairing} $\iprod{\omega^*,\omega}$ such that either
    $\iprod{\omega^*,\omega}=\omega^*(\omega)$ or $\iprod{\omega^*,\omega}=\omega(\omega^*)$.
    For reflexive spaces, this distinction does not matter,
    and for real Hilbert spaces, $\iprod{\cdot,\cdot}$ coincides with
    the inner product up to isomorphism.
    It always holds that $\iprod{\omega,\omega^*}=\iprod{\omega^*,\omega}$.
  \item The function
    $x^*\colon\mathcal{X}^*\to\overline{\R}$ as given by
    \[
    x^*(\omega^*)=\sup_{\omega\in\mathcal{X}}\bigl\{\iprod{\omega^*,\omega}-x(\omega)\bigr\}
    \]
    is the \kwd{convex conjugate} of $x$.  Similarly, the
    \kwd{convex biconjugate} $x^{**}\colon\mathcal{X}\to\overline{\R}$ is given by
    $x^{**}=(x^*)^*$, exploiting the symmetry of the dual pairing.
  \item
    The \kwd{subdifferential} of $x$ is given by:
    \[
    \partial x(\omega)=\bigl\{
    \omega^*\in\mathcal{X}^*:
    \iprod{\omega^*,\tilde{\omega}-\omega}\leq x(\tilde{\omega})-x(\omega)
    ~\text{for all $\tilde{\omega}\in\mathcal{X}$}
    \bigr\}.
    \]
  \end{itemize}
\end{definition}

\begin{lemma}\label{lemma:duality}
  \parencites[cf.][Theorems 2.3.1, 2.4.2]{Zalinescu02}\thmbreak
  Let $\mathcal{X}$ be a real topological vector space and
  $x\colon\mathcal{X}\to\overline{\R}$ be a function
  with convex conjugate $x^*\colon\mathcal{X}^*\to\overline{\R}$.
  Then:
  \begin{enumerate}
  \item[(i)]
    For all $\omega\in\mathcal{X}$ and $\omega^*\in\mathcal{X}^*$, we have
    the \kwd{Fenchel-Young inequality} which states that
    $\iprod{\omega^*,\omega}\leq x(\omega)+x^*(\omega^*)$.
  \item[(ii)]
    We have $\omega^*\in\partial x(\omega)$ iff
    $\iprod{\omega^*,\omega}= x(\omega)+x^*(\omega^*)$.
  \item[(iii)]
    For all $\omega\in\mathcal{X}$, it holds
    $x^{**}(\omega)\leq x(\omega)$.
  \end{enumerate}
\end{lemma}

\begin{proof}
  Part (i) is shown via:
  \[
  x(\omega)+x^*(\omega^*)
  =x(\omega)+\sup_{\tilde{\omega}}\bigl[\iprod{\omega^*,\tilde{\omega}}-x(\tilde{\omega})\bigr]
  \geq x(\omega)+\iprod{\omega^*,\omega}-x(\omega)
  = \iprod{\omega^*,\omega}.
  \]
  Conversely, $\omega^*\in \partial x(\omega)$, by the definition
  of the subdifferential, holds if and only if:
  \[
  \iprod{\omega^*,\tilde{\omega}-\omega}\leq x(\tilde{\omega})-x(\omega)\qquad\text{for all}\quad\tilde{\omega}\in\mathcal{X},
  \]
  and therefore equivalently:
  \[
  x(\omega)+x^*(\omega^*)
  =x(\omega)+\sup_{\tilde{\omega}}\bigl[\iprod{\omega^*,\tilde{\omega}}-x(\tilde{\omega})\bigr]
  \leq \iprod{\omega^*,\omega}.
  \]
  In combination, this gives us (ii).  Using (i) again, we find:
  \begin{align*}
    x^{**}(\omega)
    =\sup_{\tilde{\omega}^*}\bigl[\iprod{\tilde{\omega}^*,\omega}-x^*(\tilde{\omega}^*)\bigr]
    \leq\sup_{\tilde{\omega}^*}\bigl[x(\omega)+x^*(\tilde{\omega}^*)-x^*(\tilde{\omega}^*)\bigr]
    =x(\omega),
  \end{align*}
  yielding (iii).
\end{proof}

\begin{theorem}[Duality]\label{thm:duality}
  \parencite[cf.][Theorem 2.7.1]{Zalinescu02}\thmbreak
  Let $\mathcal{X},\mathcal{Y}$ be real topological vector spaces.
  Let $\Psi\colon\mathcal{X}\times\mathcal{Y}\to\overline{\R}$,
  and assume the product topology on $\mathcal{X}\times\mathcal{Y}$.
  Then, for $\gamma\in\mathcal{Y}$, we have \kwd{weak duality}:
  \[ \inf_{\tilde{\omega}}\Psi(\tilde{\omega},\gamma)
  \geq \sup_{\tilde{\gamma}^*}\bigl[\iprod{\tilde{\gamma}^*,\gamma}-\Psi^*(0,\tilde{\gamma}^*)\bigr].
  \]
  If $\mathcal{X}$ is locally convex and there exists
  $\omega\in\mathcal{X}$ such that
  $\Psi(\omega,\gamma)=\min_{\tilde{\omega}}\Psi(\tilde{\omega},\gamma)$,
  then for any $(\omega^*,\gamma^*)\in\partial\Psi(\omega,\gamma)$,
  we have \kwd{strong duality}:
  \begin{equation}\label{eq:strongduality}
    \Psi(\omega,\gamma)
  = \iprod{\gamma^*,\gamma}-\Psi^*(0,\gamma^*)
  = \max_{\tilde{\gamma}^*}\bigl[\iprod{\tilde{\gamma}^*,\gamma}-\Psi^*(0,\tilde{\gamma}^*)\bigr],
  \end{equation}
  and $\omega^*=0$.
\end{theorem}

\begin{proof}
  As linear functionals in two variables are continuous if and only if they are
  continuous in both components, we can split
  $(\mathcal{X}\times\mathcal{Y})^*=\mathcal{X}^*\times\mathcal{Y}^*$.

  We set $h(\gamma)=\inf_{\tilde{\omega}} \Psi(\tilde{\omega},\gamma)$.
  The convex conjugate can be determined as:
  \begin{align*}
    h^*(\gamma^*)
    =\sup_{\tilde{\gamma}}\bigl[\iprod{\gamma^*,\tilde{\gamma}}-\inf_{\tilde{\omega}} \Psi(\tilde{\omega},\tilde{\gamma})\bigr]
    =\sup_{\tilde{\omega},\tilde{\gamma}}\bigl[\iprod{0,\tilde{\omega}}+\iprod{\gamma^*,\tilde{\gamma}}-\Psi(\tilde{\omega},\tilde{\gamma})\bigr]
    =\Psi^*(0,\gamma^*),
  \end{align*}
  and the biconjugate is:
  \[
  h^{**}(\gamma)
  =\sup_{\tilde{\gamma}^*}\bigl[\iprod{\tilde{\gamma}^*,\gamma}-h^*(\tilde{\gamma}^*)\bigr]
  =\sup_{\tilde{\gamma}^*}\bigl[\iprod{\tilde{\gamma}^*,\gamma}-\Psi^*(0,\tilde{\gamma}^*)\bigr].
  \]
  Weak duality follows via Lemma \ref{lemma:duality}.iii.

  If $(\omega^*,\gamma^*)\in\partial\Psi(\omega,\gamma)$, then, by definition:
  \[
  \iprod{\omega^*,\tilde{\omega}-\omega}+\iprod{\gamma^*,\tilde\gamma-\gamma}\leq\Psi(\tilde\omega,\tilde\gamma)-\Psi(\omega,\gamma) \qquad\text{for all}\quad\tilde{\omega}\in\mathcal{X},\quad \tilde\gamma\in\mathcal{Y}.
  \]
  If $h(\gamma)=\Psi(\omega,\gamma)$ for some $\omega\in\mathcal{X}$, then
  $\Psi(\tilde\omega,\gamma)\geq\Psi(\omega,\gamma)$ for all
  $\tilde\omega\in\mathcal{X}$, and it follows via the Hahn-Banach theorem
  \parencite[cf.][{}3.6]{Rudin91}
  that $\omega^*=0$.  Thus, $(0,\gamma^*)\in\partial\Psi(\omega,\gamma)$
  and also $\gamma^*\in\partial h(\gamma)$.
  From Lemma \ref{lemma:duality}.ii, it then follows:
  \[
  h(\gamma)
  =\iprod{{\gamma}^*,\gamma}-h^*({\gamma}^*)
  \leq h^{**}(\gamma).
  \]
  With Lemma \ref{lemma:duality}.iii, this yields
  $h(\gamma)\leq h^{**}(\gamma)\leq h(\gamma)$ and thus
  $h(\gamma)=h^{**}(\gamma)$ with~$\gamma^*$ as a maximizer of the supremum,
  giving \eqref{eq:strongduality}.
\end{proof}

\begin{corollary}[Fenchel-Rockafellar]\label{corol:fenchel}
  \parencite[cf.][Corollary 2.8.5]{Zalinescu02}\thmbreak
  Let $\mathcal{X},\mathcal{Y}$ be real topological vector spaces and
  $f\colon\mathcal{X}\to\overline{\R}$, $g\colon\mathcal{Y}\to\overline{\R}$.
  Assume that $A\colon\mathcal{X}\to\mathcal{Y}$ is a continuous linear
  operator such that $A^*\colon\mathcal{Y}^*\to\mathcal{X}^*$ is its adjoint
  and that $\gamma\in\mathcal{Y}$ is a fixed.
  If $\mathcal{X}$ is locally convex, there exists $\omega\in\mathcal{X}$
  with
  $f(\omega)+g(A\omega-\gamma)=\min_{\tilde\omega}\bigl[f(\tilde\omega)+g(A\tilde\omega-\gamma)\bigr]$ (solving the \kwd{primal problem}),
  and $\omega^*\in\partial f(\omega)$,
  $\gamma^*\in\partial g(A\omega-\gamma)$,
  then we have strong duality with:
  \[
  f(\omega)+g(A\omega-\gamma)
  =\iprod{\gamma^*,\gamma}-f^*(A^*\gamma^*)-g^*(\gamma^*)
  =\max_{\tilde\gamma^*}\bigl[\iprod{\tilde\gamma^*,\gamma}-f^*(A^*\tilde\gamma^*)-g^*(-\tilde\gamma^*)\bigr],
  \]
  where the maximum is called the \kwd{dual problem} (right-hand side).
  Also, $-\gamma^*\in\partial g(A\omega - \gamma)$.
\end{corollary}

\begin{proof}
  We set $\Psi(\omega,\gamma)=f(\omega)+g(A\omega-\gamma)$.  Then:
  \begin{align*}
  \Psi^*(0,\gamma^*)
  &=\sup_{\omega,\gamma}\bigl[\iprod{\gamma^*,\gamma}-f(\omega)-g(A\omega-\gamma)\bigr]\\
  &=\sup_{\omega,\gamma}\bigl[\iprod{\gamma^*,A\omega-\gamma}-f(\omega)-g(\gamma)\bigr]\\
  &=\sup_{\omega,\gamma}\bigl[\iprod{A^*\gamma^*,\omega}-\iprod{\gamma^*,\gamma}-f(\omega)-g(\gamma)\bigr]\\
  &=f^*(A^*\gamma^*)+g^*(-\gamma^*).
  \end{align*}
  If $\omega^*\in\partial f(\omega)$ and
  $\gamma^*\in\partial g(A\omega-\gamma)$,
  then, for all $\tilde\omega\in\mathcal{X}$ and
  $\tilde\gamma\in\mathcal{Y}$:
  \[
  \iprod{\omega^*,\tilde\omega-\omega}
  +\iprod{\gamma^*,A\tilde\omega-\tilde\gamma-A\omega+\gamma}
  \leq
  f(\tilde\omega)-f(\omega)+g(A\tilde\omega-\tilde\gamma)-g(A\omega-\gamma),
  \]
  and so
  $(\omega^*+A^*\gamma^*,-\gamma^*)\in\partial\Psi(\omega,\gamma)$.
  We can now apply Theorem \ref{thm:duality} to obtain strong duality,
  and it follows that $\omega^*+A^*\gamma^*=0$.  Therefore:
  \[
  \iprod{-\gamma^*,\tilde\gamma-\gamma}\leq
  g(A\omega-\tilde\gamma)-g(A\omega-\gamma),
  \]
  and thus $-\gamma^*\in\partial g(A\omega-\gamma)$.
\end{proof}

\subsection{Duality on Beurling LASSO}

For \eqref{eq:blasso}, we can choose:
\[
f(\nu)=\alpha\snorm{\nu}_{\mathrm{TV}},\qquad
g(\gamma)=\frac{1}{2}\snorm{\gamma}_{\mathcal{H}}^2,\qquad
\alpha>0,\quad \gamma=A\nu-b,
\]
in order to apply Corollary \ref{corol:fenchel}.
First we have to show that the minimum is attained.
For the primal problem, this was done by
\textcite[Proposition~3.1]{Bredies13} via the \kwd{direct method}.
We now reenact the proof with some detail added in.  We begin with some
well-known statements from functional and convex analysis:

\begin{lemma}\label{lemma:adjoint}
  \parencite[cf.][Proposition VI.1.3]{Conway90}\thmbreak
  Let $\mathcal{X},\mathcal{Y}$ be topological vector spaces and
  let $B\colon\mathcal{Y}\to\mathcal{X}$ be a continuous linear
  operator.  Then its adjoint $B^*\colon\mathcal{X}'\to\mathcal{Y}'$ is
  weak*-weak*-continuous.
\end{lemma}

\begin{proof}
  By the definition of the adjoint, we have, for any $\omega^*\in\mathcal{X}'$
  and $\gamma\in\mathcal{Y}$:
  \[
  \iprod{B\gamma,\omega^*}
  =\iprod{\gamma,B^*\omega^*}.
  \]
  Now, as $\omega^*$ converges in the weak* topology over
  $\mathcal{X}'$, then
  $\iprod{B\gamma,\omega^*}$ converges in $\R$, so
  $B^*\omega^*$ converges in the weak* topology over $\mathcal{Y}'$.
\end{proof}

\begin{lemma}\label{lemma:iso}
  \parencites[cf.][Section~6.8]{Aliprantis06}[Proposition~VI.1.4]{Conway90}\thmbreak
  Let $\mathcal{H}$ be a real Hilbert space and let $\mathcal{X}$ be a
  real Banach space.  Let $A\colon\mathcal{X}'\to\mathcal{H}$ and
  $B\colon\mathcal{H}\to\mathcal{X}$ be continuous linear operators
  such that $B^*=A$.  For any $\omega^*\in\mathcal{X}'$ and
  $\gamma\in\mathcal{H}$, it holds that
  $\omega^*(B\gamma)=A^*\gamma(\omega^*)$.
\end{lemma}

\begin{proof}
  Since $\mathcal{H}$ is a real Hilbert space, we have for any
  $\gamma\in\mathcal{H}$ that $\gamma=\iprod{\gamma,\cdot}\in\mathcal{H}^*$,
  and therefore, with $\omega^*\in\mathcal{X}$:
  \[
  \iprod{B\gamma,\omega^*}
  =\iprod{\gamma,B^*\omega^*}
  =\iprod{\gamma,A\omega^*}
  =\iprod{A^*\gamma,\omega^*}.
  \]
\end{proof}

\begin{lemma}
  \parencite[cf.][Lemma~6.22]{Aliprantis06}\thmbreak
  Let $\mathcal{X}$ be a real Banach space and let $\mathcal{X}'$ be its dual
  space with the norm~$\snorm{\cdot}_{\mathcal{X}'}$.
  Then $\snorm{\cdot}_{\mathcal{X}'}$ is weak* lower semicontinuous.
\end{lemma}

\begin{proof}
  When we apply the definition and assume that the neighborhood
  $V\subset\mathcal{X}'$ is always open in the weak* topology,
  we have:
  \begin{align*}
  \liminf_{\substack{\tilde\omega^*\to\omega^*\\\text{weak*}}} \snorm{\tilde\omega^*}_{\mathcal{X}'}
  &=\sup_{V\ni\omega^*}\inf_{\substack{\tilde\omega^*\in V\\\tilde\omega^*\neq\omega^*}}
   \snorm{\tilde\omega^*}_{\mathcal{X}'}\\
  &=\sup_{V\ni\omega^*}\inf_{\substack{\tilde\omega^*\in V\\\tilde\omega^*\neq\omega^*}}\sup_{\snorm{\tilde\omega}=1}
   \iprod{\tilde\omega^*,\tilde\omega}\\
  &\geq \sup_{V\ni\omega^*}\sup_{\snorm{\tilde\omega}=1}\inf_{\substack{\tilde\omega^*\in V\\\tilde\omega^*\neq\omega^*}}
   \iprod{\tilde\omega^*,\tilde\omega}\\
  &= \sup_{\snorm{\tilde\omega}=1}\sup_{V\ni\omega^*}\inf_{\substack{\tilde\omega^*\in V\\\tilde\omega^*\neq\omega^*}}
   \iprod{\tilde\omega^*,\tilde\omega}\\
  &= \sup_{\snorm{\tilde\omega}=1}\lim_{\substack{\tilde\omega^*\to\omega^*\\\text{weak*}}}
   \iprod{\tilde\omega^*,\tilde\omega}\\
  &= \sup_{\snorm{\tilde\omega}=1}
   \iprod{\omega^*,\tilde\omega}\\
  &= \snorm{\omega^*}_{\mathcal{X}'}.
  \end{align*}
\end{proof}

\begin{lemma}\label{lemma:mincompact}
  \parencite[cf.][Theorem 2.43]{Aliprantis06}\thmbreak
  Let $\mathcal{X}'$ be a dual real Banach space.  Let
  $f\colon\mathcal{X}'\to\R$ be a weak* lower semicontinuous function.
  If $U\subset\mathcal{X}'$ is a weak* compact set, then $f$ attains its
  minimum on $U$.
\end{lemma}

\begin{proof}
  We show this by contradiction.  Assume that the infimum
  $a\coloneqq\inf_{\omega^*\in U}f(\omega^*)$ is not
  attained.  Then, for any $\omega^*\in U$, since $f$ is weak* lower
  semicontinuous, there exists
  a weak* open neighborhood $V_{\omega^*}\subset\mathcal{X}'$
  with $\omega^*\in V_{\omega^*}$ such that
  $f(\tilde\omega^*)\geq (a+f(\omega^*))/2$ for all
  $\tilde\omega^*\in V_{\omega^*}$.  Then $\cup_{\omega^*\in U}V_{\omega^*}$
  is an open cover of the compact set $U$, so there exist
  $\omega_1^*,\dotsc,\omega_m^*\in U$ such that
  $U=\cup_{k=1}^mV_{\omega^*_k}$.  Therefore, there is a
  $k\in\{1,\dotsc,m\}$ such that
  $\inf_{\tilde\omega^*\in V_{\omega^*_k}}f(\tilde\omega^*)\leq a$.
  However, since $f(\tilde\omega^*)\geq (a+f(\omega_k^*))/2>a$ for all
  $\tilde\omega^*\in V_{\omega_k^*}$, this is impossible.
\end{proof}

\begin{corollary}\label{corol:min}
  Let $\mathcal{X}'$ be a dual real Banach space.  Let
  $f\colon\mathcal{X}'\to\R$ be a weak* lower semicontinuous and coercive
  function.  Then $f$ attains its minimum on $\mathcal{X}'$.
\end{corollary}

\begin{proof}
  Since $f$ is coercive, for any constant $C>0$, there exists a value $M>0$
  such that if $f(\omega')\leq C$ with $\omega'\in\mathcal{X}'$, then
  $\snorm{\omega'}_{\mathcal{X}'}\leq M$.
  According to the Banach-Alaoglu theorem\linebreak{}
  \parencites[cf.][{}3.15]{Rudin91}[Theorem~V.3.1]{Conway90}, the set
  \[
  U=\{\omega'\in\mathcal{X}':\snorm{\omega'}_{\mathcal{X}'}\leq M\}
  \]
  is weak* compact.  According to Lemma \ref{lemma:mincompact},
  $f$ thus attains its minimum on $U\subset\mathcal{X}'$.
\end{proof}

Assume that there exists a continuous linear operator
$B\colon \mathcal{H}\to C_0(\R)$ such that $A=B^*$.  With
Lemma~\ref{lemma:adjoint}, it follows that $A$ is weak*-weak*-continuous,
and it is also bounded.  Since $f$ and $g$ are both
composed of norms, this means that \eqref{eq:blasso} is coercive and
weak* lower semicontinuous in $\nu$.  Via Corollary \ref{corol:min},
the minimum is attained.

For any $\gamma\in\mathcal{H}$, we set $\gamma^*=\iprod{\gamma,\cdot}$
(in the sense of the inner product), so we have:
\[
\iprod{\gamma^*,\tilde\gamma-\gamma}=\iprod{\gamma,\tilde\gamma-\gamma}\leq
\frac{1}{2}\snorm{\tilde\gamma}^2_{\mathcal{H}}-\frac{1}{2}\snorm{\gamma}^2_{\mathcal{H}}\qquad
\text{for all}\quad \tilde\gamma\in\mathcal{H}.
\]
Thus, $\gamma^*\in\partial g(\gamma)$; in fact,
$\partial g(\gamma)=\{\gamma^*\}$ since if $\hat\gamma^*=\iprod{\hat\gamma,\cdot}$
with $\hat\gamma\neq\gamma$, we then have:
\[
0<\frac{1}{2}\snorm{\hat\gamma-\gamma}_{\mathcal{H}}^2
=-\frac{1}{2}\snorm{\hat\gamma}_{\mathcal{H}}^2+\iprod{\hat\gamma,\hat\gamma-\gamma}
+\frac{1}{2}\snorm{\gamma}_{\mathcal{H}}^2,
\]
so $\hat\gamma^*\not\in\partial g(\gamma)$.

Considering $f$, we know \parencite[cf.][{}6.12]{Rudin87}
that for any $\nu\in\mathcal{M}(\R)$, there exists a Borel-measurable
function $u\colon \R\to\{-1,1\}$ from which we can construct a linear functional
$\nu^*\in\mathcal{M}'(\R)$ with
$\nu^*(\nu)=\int u\:\dif\nu=\snorm{\nu}_{\mathrm{TV}}$.  Then:
\[
\iprod{\nu^*,\tilde\nu-\nu}
=\int u\:\dif(\tilde\nu-\nu)
=\int u\:\dif\tilde\nu-\int u\:\dif\nu
\leq \snorm{\tilde\nu}_{\mathrm{TV}}-\snorm{\nu}_{\mathrm{TV}}\qquad
\text{for all}\quad \tilde\nu\in\mathcal{M}(\R),
\]
and therefore $\alpha\nu^*\in\partial f(\nu)$.  We can now
apply Corollary \ref{corol:fenchel} in order to obtain strong duality.

Even though generally $\partial f(\nu)\not\subseteq C_0(\R)$, we can
identify $A^*\gamma^*=B\gamma^*$ according to Lemma~\ref{lemma:iso},
and therefore it is sufficient to
regard $f^*\colon C_0(\R)\to\overline{\R}$ in order to interpret the result
of Corollary~\ref{corol:fenchel}.  We compute:
\begin{align*}
  f^*(\nu^*)
  &=\sup_{\nu\in\mathcal{M}(\R)}\bigl[\iprod{\nu^*,\nu}-f(\nu)\bigr]\\
  &=\sup_{\nu\in\mathcal{M}(\R)}\bigl[\iprod{\nu^*,\nu}-\alpha\snorm{\nu}_{\mathrm{TV}}\bigr]\\
  &=
  \begin{cases}
    0,&\text{for $\snorm{\nu^*}_{\infty}\leq\alpha$},\\
    \infty,&\text{otherwise}
  \end{cases}\\
  &=\imath_{\snorm{\cdot}_\infty\leq\alpha}(\nu^*),
\end{align*}
where $\imath_{\snorm{\cdot}_\infty\leq\alpha}$ is the \kwd{indicator function},
since, according to the Hahn-Banach theorem, if $\nu^*\neq0$, then
there exists $\nu\in\mathcal{M}(\R)$ such that
$\iprod{\nu^*,\nu}=\snorm{\nu^*}_{\infty}\snorm{\nu}_{\mathrm{TV}}$
becomes arbitrarily large.
For the conjugate of $g$, we have:
\begin{align*}
  g^*(\gamma^*)
  &=\sup_{\gamma\in\mathcal{H}}\biggl[\iprod{\gamma^*,\gamma}-\frac{1}{2}\snorm{\gamma}_{\mathcal{H}}^2\biggr]\\
  &=\sup_{\gamma\in\mathcal{H}}\biggl[\frac{1}{2}\snorm{\gamma^*}_{\mathcal{H}}^2-\frac{1}{2}\snorm{\gamma^*-\gamma}_{\mathcal{H}}^2\biggr]\\
  &=\frac{1}{2}\snorm{\gamma^*}_{\mathcal{H}}^2.
\end{align*}
We can thus formulate the dual problem as:
\begin{equation}\label{eq:blassodual}
\max_{r\in\mathcal{H}}\biggl\{\iprod{r,b}-\frac{1}{2}\snorm{r}_{\mathcal{H}}^2:
\snorm{A^*r}_{\infty}\leq \alpha\biggr\},
\end{equation}
where we have $r=b-A\nu$ due to $\partial g(A\nu-b)=\{b-A\nu\}$.
In other words, the solution of the dual problem is nothing but the
\emph{residual} of the primal problem.  In some applications like
denoising, it could potentially be sufficient to know $A\nu$ while
avoiding stating $\nu$ directly.  Also, the benefit of
solving $A\nu=b-r$ rather than \eqref{eq:blasso} is that it is
only a linear equation and no longer an optimization problem.
This property is exploited by \textcite{Catala17} in a
semidefinite relaxation approach.

While the objective of the dual problem \eqref{eq:blassodual}
is linear and quadratic, its constraint still involves the global absolute
maximum of a function $A^*r\in C_0(\R)$.  However, if $\mathcal{H}$ is discrete,
then knowledge about the structure of $A^*$ can be used to predict a
neighborhood of the maximum.  \textcite[Algorithm 1]{Catala17} again propose
a Frank-Wolfe-type algorithm with~BFGS.

\section{Application to Source Separation}
\label{sec:blassosep}

Conceptually speaking, continuous LASSO is always a hard problem,
and even Beurling LASSO cannot eliminate the difficulty.  However, it gives
a powerful framework in order to analyze the problem in other ways.
In source separation, we can avoid parametrizing the Radon measure $\nu$ by
using an \emph{intermediate representation} instead.

As an illustrative example, let us consider two patterns $y_1,y_2$,
where $y_1$ is the upper half of an ellipse and $y_2$ is
triangular-shaped.  Giving a mixture spectrum, the task is to separate
the contributions of the individual patterns.

In Figure \ref{fig:blasso}, the different stages of representation are
displayed.  The left plot is the complete mixture spectrogram with
the contributions of both patterns.  In the middle column, these
contributions are separated.  The plots in the right column are
linear combinations of shifted Dirac measures (indicated as arrows).
Convolving the spectra in the right column with the respective
patterns gives the spectra in the middle column.

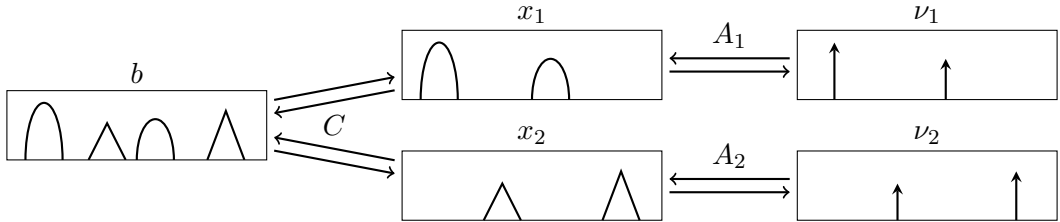
\begin{figure}[htp]
  \centering
  \begin{tikzpicture}
    \begin{axis}[name=mix,at={(-5.2cm,0)},anchor=center,ymin=0,ymax=1.7,xmin=0,xmax=7,xtick=\empty,ytick=\empty,title={$b$},
      width=5cm,height=2.5cm,title style={yshift=-1.5ex}]
      \draw[thick] (axis cs:1,0) ellipse [x radius=0.5,y radius=1.4];
      \draw[thick] (axis cs:4,0) ellipse [x radius=0.5,y radius=1];
      \addplot[thick] coordinates {(2.2,0) (2.7,0.9) (3.2,0)};
      \addplot[thick] coordinates {(5.4,0) (5.9,1.2) (6.4,0)};
    \end{axis}
    \begin{axis}[name=rect,at={(0cm,0.8cm)},anchor=center,ymin=0,ymax=1.7,xmin=0,xmax=7,xtick=\empty,ytick=\empty,
      width=5cm,height=2.5cm,title={$x_1$},title style={yshift=-1.5ex}]
      \draw[thick] (axis cs:1,0) ellipse [x radius=0.5,y radius=1.4];
      \draw[thick] (axis cs:4,0) ellipse [x radius=0.5,y radius=1];
    \end{axis}
    \begin{axis}[name=tri,at={(0cm,-0.8cm)},anchor=center,ymin=0,ymax=1.7,xmin=0,xmax=7,xtick=\empty,ytick=\empty,
      width=5cm,height=2.5cm,title={$x_2$},title style={yshift=-1.5ex}]
      \addplot[thick] coordinates {(2.2,0) (2.7,0.9) (3.2,0)};
      \addplot[thick] coordinates {(5.4,0) (5.9,1.2) (6.4,0)};
    \end{axis}
    \begin{axis}[name=recta,at={(5.2cm,0.8cm)},anchor=center,ymin=0,ymax=1.7,xmin=0,xmax=7,xtick=\empty,ytick=\empty,
      width=5cm,height=2.5cm,title={$\nu_1$},title style={yshift=-1.5ex}]
      \draw[thick,->,>=stealth] (axis cs:1,0) -- (axis cs:1,1.4);
      \draw[thick,->,>=stealth] (axis cs:4,0) -- (axis cs:4,1);
    \end{axis}
    \begin{axis}[name=tria,at={(5.2cm,-0.8cm)},anchor=center,ymin=0,ymax=1.7,xmin=0,xmax=7,xtick=\empty,ytick=\empty,
      width=5cm,height=2.5cm,title={$\nu_2$},title style={yshift=-1.5ex}]
      \draw[thick,->,>=stealth] (axis cs:2.7,0) -- (axis cs:2.7,0.9);
      \draw[thick,->,>=stealth] (axis cs:5.9,0) -- (axis cs:5.9,1.2);
    \end{axis}
    \draw[thick,->,transform canvas={yshift=2.5pt},shorten <=1mm,shorten >=1mm] ($(mix.north east)!0.5!(mix.east)$) -- ($(rect.south west)!0.5!(rect.west)$);
    \draw[thick,<-,transform canvas={yshift=-2.5pt},shorten <=1mm,shorten >=1mm] ($(mix.north east)!0.5!(mix.east)$) -- ($(rect.south west)!0.5!(rect.west)$);
    \draw[thick,<-,transform canvas={yshift=2.5pt},shorten <=1mm,shorten >=1mm] ($(mix.south east)!0.5!(mix.east)$) -- ($(tri.north west)!0.5!(tri.west)$);
    \draw[thick,->,transform canvas={yshift=-2.5pt},shorten <=1mm,shorten >=1mm] ($(mix.south east)!0.5!(mix.east)$) -- ($(tri.north west)!0.5!(tri.west)$);
    \draw[thick,<-,transform canvas={yshift=2.5pt},shorten <=1mm,shorten >=1mm] (rect) -- node[above] {$A_1$} (recta);
    \draw[thick,->,transform canvas={yshift=-2.5pt},shorten <=1mm,shorten >=1mm] (rect) -- (recta);
    \draw[thick,<-,transform canvas={yshift=2.5pt},shorten <=1mm,shorten >=1mm] (tri) -- node[above] {$A_2$} (tria);
    \draw[thick,->,transform canvas={yshift=-2.5pt},shorten <=1mm,shorten >=1mm] (tri) -- (tria);
    \node at (-2.6,0) {$C$};
  \end{tikzpicture}
    \caption{Separation of the contributions of two different patterns
      in a spectrum}
    \label{fig:blasso}
\end{figure}

To formalize this process, we have to extend our framework.
While the spectra in the right column of Figure \ref{fig:blasso}
can be understood as a Radon measure each, the space $\mathcal{M}(\R)$
only accounts for one measure, not multiple ones.  Thus, to
operate with multiple patterns, we have to consider
$\mathcal{M}(\R)^n$ which is then the dual
space of $C_0(\R)^n$, where $n$ is the number of patterns
($n=2$ in the figure).
When we equip the latter with the norm
\[
\snorm{z}_{\infty}
=\max_{i=1,\dotsc,n}\snorm{z_i}_{\infty},\qquad
z=(z_1,\dotsc,z_n)\in C_0(\R)^n
\]
(which is compatible with the product topology), it gives the dual norm:
\[
\snorm{\nu}_{\mathrm{TV}}
=\sum_{i=1}^n\snorm{\nu_i}_{\mathrm{TV}},\qquad
\nu=(\nu_1,\dotsc,\nu_n)\in\mathcal{M}(\R)^n,
\]
complying with \textcite{Bredies13}.  For the Hilbert space
$\mathcal{H}^n$, we use:
\[
\snorm{x}_{\mathcal{H}^n}^2
=\sum_{i=1}^n\snorm{x_i}_{\mathcal{H}}^2,\qquad
x=(x_1,\dotsc,x_n)\in\mathcal{H}^n.
\]

Following Figure \ref{fig:blasso}, the operator
$A\colon\mathcal{M}(\R)^n\to\mathcal{H}^n$ now convolves
the measure $\nu\in\mathcal{M}(\R)^n$ component-wise with
the patterns $y_1,\dotsc,y_n\in C_0(\R)$:
\[
A\nu=\begin{pmatrix}
  A_1\nu_1\\\vdots\\A_n\nu_n
  \end{pmatrix}=
  \begin{pmatrix}
  \nu_1*y_1\\\vdots\\\nu_n*y_n
  \end{pmatrix}\eqqcolon x,
\]
where the convolution is defined via:
\[
(\nu*y)(\omega)=\int y(\omega-s)\:\dif \nu(s).
\]
The space $\mathcal{H}$ and the patterns $y_1,\dotsc,y_n$
have to be chosen such that the pre-adjoint operator is well-defined,
that is, there exists a continuous linear operator
$B\colon \mathcal{H}^n\to C_0(\R)^n$ such that
$A=B^*$.

The operator $C\colon\mathcal{H}^n\to\mathcal{H}$ sums the components
of the individual patterns:
\[
Cx=\sum_{k=1}^n x_k,
\]
and it is obviously linear and continuous.  Combined, we formulate the
primal problem as:
\begin{equation}\label{eq:blassosep}
  \min_{x,\nu}\biggl\{\frac{1}{2}\snorm{Cx-b}_{\mathcal{H}}^2+\alpha\snorm{\nu}_{\mathrm{TV}}:A\nu=x\biggr\}.
\end{equation}
so we have:
\begin{align*}
f(x)&=\min_{\nu}\bigl[\alpha\snorm{\nu}_{\mathrm{TV}}+\imath_0(A\nu-x)]
,\qquad \alpha>0,
\end{align*}
with:
\begin{align*}
f^*(x^*)
&=\sup_{x}\bigl[\iprod{x^*,x}-f(x)\bigr]\\
&=\sup_{\nu,x}\bigl[\iprod{x^*,x}-\alpha\snorm{\nu}_{\mathrm{TV}}-\imath_0(A\nu-x)\bigr]\\
&=\sup_{\nu,x}\bigl[\iprod{x^*,A\nu-x}-\alpha\snorm{\nu}_{\mathrm{TV}}-\imath_0(x)\bigr]\\
&=\sup_{\nu}\bigl[\iprod{A^*x^*,\nu}-\alpha\snorm{\nu}_{\mathrm{TV}}\bigr]\\
&=\imath_{\snorm{\cdot}\leq\alpha}(A^*x^*).
\end{align*}
It would now be straight-forward to apply Corollary \ref{corol:fenchel} again,
but it would still only give the residual, not expose $x$ directly.
However, unlike the original problem \eqref{eq:blasso}, the new problem
\eqref{eq:blassosep} is now one where the solution $x\in\mathcal{H}^n$
lies in a Hilbert space and only the constraint is problematic.

Just like normal LASSO is often solved via the \kwd{proximal mapping}
\parencite[cf.][Section 1.1]{Parikh14}, we can
formulate the proximal mapping for \eqref{eq:blassosep} and apply
Corollary \ref{corol:fenchel} on it:
\begin{align*}
  \operatorname{prox}_f(x)&=\argmin_{\tilde{x}\in\mathcal{H}^n}\biggl[\frac{1}{2}\snorm{\tilde{x}-x}_{\mathcal{H}^n}^2+f(\tilde{x})\biggr]\\
  &=x-\argmax_{x^*\in\mathcal{H}^n}\biggl[\iprod{x^*,x}-\frac{1}{2}\snorm{x^*}_{\mathcal{H}^n}^2-f^*(x^*)\biggr]\\
  &=x-\argmax_{x^*\in\mathcal{H}^n}\biggl[\frac{1}{2}\snorm{x}_{\mathcal{H}^n}^2-\frac{1}{2}\snorm{x^*-x}_{\mathcal{H}^n}^2-f^*(x^*)\biggr]\\
  &=x-\operatorname{prox}_{f^*}(x),
\end{align*}
where we set:
\[
g(x)=\frac{1}{2}\snorm{x}^2_{\mathcal{H}^n},\qquad\text{so}\quad
g^*(x^*)=\frac{1}{2}\snorm{x^*}^2_{\mathcal{H}^n}.
\]
This result is also known as
\kwd{Moreau decomposition} \parencite[cf.][Section 2.5]{Parikh14}.
When substituting $\tilde{x}=A\nu$, the primal
problem here is formally equivalent to \eqref{eq:blasso}, so an optimal
$\nu$ exists, and therefore also an optimal $\tilde{x}$.
The proximal gradient iteration \parencite[cf.][Section~4.2]{Parikh14}
for \eqref{eq:blassosep} is then:
\begin{equation}
\begin{aligned}\label{eq:blassoiter}
  x^{i+1}&=\operatorname{prox}_{\lambda f}\bigl(x^i-\lambda\,(Cx^i-b)\bigr)\\
  &=x^i-\lambda\,(Cx^i-b)-\operatorname{prox}_{(\lambda f)^*}\bigl(x^i-\lambda\,(Cx^i-b)\bigr)\\
  &=x^i-\lambda\,(Cx^i-b)-\argmax_{x^*\in\mathcal{H}^n}
\biggl\{\frac{1}{2}\bignorm{x^*-x^i+\lambda\,(Cx^i-b)}_{\mathcal{H}^n}^2
:\snorm{A^*x^*}_\infty\leq\lambda\alpha\biggr\},
\end{aligned}
\end{equation}
with $\lambda>0$.

So far, we have not specified the choice of the Hilbert space $\mathcal{H}$.
With $\nu_i\in\mathcal{M}(\R)$, $x_i\in\mathcal{H}$, and $i=1,\dotsc,n$,
we have:
\begin{equation*}
\iprod{A_i\nu_i,x_i}
=\int (\nu_i*y_i)(\omega)\, x_i(\omega)\:\dif\omega
=\int\int y_i(\omega -s)\:\dif\nu_i(s)\, x_i(\omega)\:\dif\omega.
\end{equation*}
For the pre-adjoint operator to exist, we need to be able to
swap the integrals.
If $\mathcal{H}=L_2(\R)$, then this is well-defined for
$y_i\in C_0(\R)\cap L_2(\R)$:  As can be shown by applying a version of the
convolution theorem \parencite[cf.][Theorem~2.5.9.a]{Benedetto96}
in combination with Riemann-Lebesgue lemma
\parencites[cf.][{}7.5]{Rudin91}[Theorem 1.4.1.c]{Benedetto96},
the function given by $A^*x_i(s)=\int y_i(\omega -s)\,x_i(\omega)\:\dif\omega$
then lies in $C_0(\R)$ as well.

For computations, it is practical to choose a discrete Hilbert space
such as $\mathcal{H}=\ell_2(\Z)$.  In this case, we need to ensure sufficient
decay of the patterns even when they are sampled.  A~possible choice
is $y_i\in C_0(\R)\cap W(\R)$ (where $W(\R)$ is the \kwd{Wiener space}
\parencite[cf.][Definition~6.1.1]{Groechenig01}),
yielding $A^*x_i\in C_0(\R)\cap L_2(\R)$.
Note that discretizing $x_i\in\mathcal{H}$ does not restrict the space
for~$\nu_i$; however, if the grid is too coarse, then some features of
$y_i$ may disappear between the sampling points.

\section{Conclusion}

With \eqref{eq:blassosep}, we have given an explicit proximal gradient
iteration in order to separate the convolutional contributions of
given patterns from a mixture.  Implicitly, it solves the continuous
problem \eqref{eq:blasso}, but by avoiding representing the measures
directly, the computation can be carried out in a discrete Hilbert
space.

Even though the linear operator $A$ is a convolution in our example,
the formulation is not limited to convolutions as long as the
pre-adjoint operator can be stated.  However, the caveat is that
the bounds of $A^*x^*$ give constraints over a continuous function.
How to incorporate those in a practical solution algorithm is yet to be
determined.

\section*{Acknowledgements}

The first author acknowledges funding by the Deutsche Forschungsgemeinschaft
  (DFG, German Research Foundation) -- Projektnummer 281474342/GRK2224/1.

\renewcommand*{\bibfont}{\raggedright}

\setlength\bibitemsep{.5\itemsep}

\printbibliography[heading=bibintoc]
\end{document}